\documentclass[11pt]{article} 

\usepackage{geometry} 
\usepackage{graphicx, amsmath,amsfonts,cite,appendix, amssymb,xcolor,subfig, enumerate, mathrsfs, amsthm,url,mathtools} 
\usepackage[pdfencoding=auto,psdextra]{hyperref}
\usepackage{epsfig}
\usepackage{booktabs} 
\usepackage{array} 
\usepackage{paralist} 
\usepackage{verbatim} 
\usepackage{subfig} 

\newcommand{\be}{\begin{equation}}
\newcommand{\ee}{\end{equation}}
\newcommand{\ben}{\begin{equation*}}
\newcommand{\een}{\end{equation*}}
\newcommand{\mc}{\mathcal}
\newcommand{\mbf}{\mathbf}

\newtheorem{lemma}{Lemma}[section]
\newtheorem{defi}{Definition}[section]
\newtheorem{thm}{Theorem}

\newtheorem{rem}{Remark}[section]
\newtheorem{prop}[lemma]{Proposition}

\newcommand{\e}{\epsilon}

\newcommand{\abs}[1]{\left\lvert#1\right\rvert}
\newcommand{\norm}[1]{\left\lVert#1\right\rVert}
\newcommand{\expec}{\mathbb{E}}

\DeclareMathOperator*{\argmin}{arg\,min}

\DeclareMathOperator*{\esssup}{ess\,sup}

\newcommand{\ep}{{\mathbb {E}}}
\newcommand{\II}{{\mathbb {I}}}
\newcommand{\pr}{{\mathbb {P}}}
\newcommand{\re}{{\mathbb{R}}}
\newcommand{\vc}[1]{{\mathbf #1}}
\newcommand{\vd}[1]{{\boldsymbol #1}}
\newcommand{\ol}[1]{{\overline{#1}}}
\newcommand{\whth}{\widehat{\theta}}
\newcommand{\whG}{\widehat{G}}
\newcommand{\whi}{\hat{i}}
\newcommand{\whv}{\hat{V}}
\newcommand{\whx}{\hat{\vc{x}}}
\newcommand{\ftau}{f_{\vd{\tau}}}
\newcommand{\MM}{{\mathcal{M}}}
\newcommand{\FF}{{\mathcal{F}}}
\newcommand{\AAA}{{\mathcal{A}}}
\newcommand{\pack}{{\mathcal{P}}_{M, d_{\min}}}
\newcommand{\packval}[1]{{\mathcal{P}}_{M, #1}}
\newcommand{\packvalp}[1]{{\mathcal{P}}_{M', #1}}
\newcommand{\packvalx}[1]{{\mathcal{P}}_{M+1, #1}}
\newcommand{\minimaxcs}{\mathsf{M}^*(\mbf{A})}
\newcommand{\trace}{\textsf{Tr}}

\usepackage{sectsty}
\allsectionsfont{\sffamily\mdseries\upshape} 

\title{\vspace{-20pt}A strong converse bound for multiple hypothesis testing, \\ with applications to high-dimensional estimation \vspace{-10pt}}
\date{\today}
\author{Ramji Venkataramanan\thanks{Department of Engineering, University of Cambridge, Trumpington Street, Cambridge CB3 0DZ, UK. Email: {\tt ramji.v@eng.cam.ac.uk}}  \and 
Oliver Johnson\thanks{School of Mathematics, University of Bristol, University Walk, Bristol, BS8 1TW, UK. Email: {\tt maotj@bristol.ac.uk}} }
\begin{document}

\maketitle

\begin{abstract}
In  statistical inference problems, we wish to obtain lower bounds on the minimax risk, that is to bound  the performance of any possible estimator. A standard technique to do this involves the use of Fano's inequality. However, recent work in an information-theoretic setting has shown  that an argument based on binary hypothesis testing gives tighter converse results (error lower bounds) than Fano for channel coding problems. We adapt this technique to the statistical setting, and argue that Fano's inequality can always be replaced by this approach to obtain tighter lower bounds that can be easily computed and are asymptotically sharp. We illustrate our technique in three  applications: density estimation, active learning of a binary classifier, and compressed sensing,  obtaining tighter risk lower bounds in each case.
\end{abstract}
 
 \section{Introduction} \label{sec:intro}

When solving an inference problem, we would like to know if the algorithm we use is close to optimal. In statistical language we seek to give a lower  bound on the performance
  of any  estimator over a class of problems (often called the minimax risk over the class).
In  the language of information theory, we speak of  converse results, which  give  performance bounds for all
communication  schemes  over a  noisy channel.

In the statistics literature, one  standard approach to proving converse results is via Fano's inequality (see \cite[Theorem 2.11.1]{cover}). 
However, recent  information-theoretic literature has shown how to obtain sharper converse bounds. The resulting improvements can be significant at finite sample size, and give bounds that are close to optimal, as illustrated in the work of Polyanskiy, Poor and Verd\'{u} \cite{ppv10}.
The present paper shows how the method of \cite{ppv10}, although developed for channel coding problems,  gives stronger risk lower bounds  for high-dimensional estimation problems, compared to  the standard Fano approach.

We first describe the  general set-up, following the treatment and notation of \cite[Chapter 2]{tsybakov09Book}.
Consider an inference problem (possibly non-parametric) where we wish to estimate some  $\theta \in \FF$ from samples $\vc{Y} = (Y_1, \ldots, Y_n)$ generated according to a distribution $P_{\theta}(\vc{Y})$. For example, 
in Section \ref{sec:densityest} we consider $\theta$ to be a probability density chosen from a pre-specified class, and
in Section \ref{sec:cs_bound} we consider $\theta$ to be a $k$-sparse vector in $\re^n$. Let $\whth := \whth(\vc{Y})$ be any estimator of $\theta$ and let $d(\theta, \whth)$ represent the loss. We assume that $d$ is a distance, although (as in \cite{tsybakov09Book}) our results hold when $d$ is a semi-distance; that is,  when $d(\theta, \theta') = 0$ need not imply that $\theta = \theta'$.   We obtain lower bounds on  the minimax risk
\begin{equation} \label{eq:tolowerbound}
\inf_{\whth} \sup_{\theta \in \FF} \, \ep\left[ w( d(\theta, \whth) ) \right],
\end{equation}
where $w$ is any monotonically increasing function with $w(0) = 0$. For example, we may consider $w(u) = u^p$ for any $p > 0$ or $w(u) = \II( u \geq c)$ for some $c > 0$.

  A standard approach for obtaining a lower bound on \eqref{eq:tolowerbound} is as follows. First, a set $\{ \theta_1, \ldots, \theta_M \}$ $\subseteq \FF$ is chosen, with a lower bound on the pairwise distance between any two of its elements, where the distance is measured using the loss function $d(\cdot, \cdot)$. Then,  any estimator $\widehat{\theta}$  defines an $M$-ary hypothesis test that  detects one of $\{ \theta_1, \ldots, \theta_M \}$ based on the data $\mbf{Y}$. Next, the key step is to obtain a lower bound for the error probability associated with this hypothesis test. For a well-constructed set, Fano's inequality  often shows that this average error probability  is bounded away from $0$ as $n \to \infty$. In this paper, we present a technique that often shows that  it  approaches $1$ as $n \to \infty$. In information theory parlance (see for example \cite[P.207]{cover}), we prove a ``strong converse" result in contrast to the ``weak converse" provided by Fano's inequality. 

We now explain the details, following the  framework in \cite{tsybakov09Book} (see also \cite{gine}, \cite{massart2007}). For any positive constants $A$ and 
$\psi_n$, using Markov's inequality we have
\begin{align*}
\pr \left( d( \theta, \whth) \geq A \psi_n \right) 
& =  \pr \left( w \left( \frac{1}{\psi_n} d(\theta, \whth) \right) \geq w(A) \right) 
 \leq  \frac{ \ep \left[ w \left( \frac{1}{\psi_n} d(\theta, \whth) \right) \right]}{w(A)},
\end{align*}
which implies 
\begin{equation} \label{eq:tsyb}
\sup_{\theta \in \FF} \ep \left[ w \left( \frac{1}{\psi_n} d( \theta, \whth) \right) \right] \geq w(A) \left( \sup_{\theta \in \FF} \pr \left( d( \theta, \whth) \geq A \psi_n \right) \right).
\end{equation}
When applying \eqref{eq:tsyb}, we typically choose $\psi_n$ as a decreasing function of $n$ to give the desired convergence rate, and $A$ as a constant that can be used to optimize the lower bound.
The goal then is  to  control the bracketed term on the RHS of \eqref{eq:tsyb}  to obtain a lower bound on the minimax risk.  We use the following definition.
\begin{defi} A collection $\pack = \{ \theta_1, \ldots, \theta_M \} \subseteq \FF$ is called a packing set of size $M$ and minimum distance $d_{\min}$ if
$$ d( \theta_i, \theta_j ) \geq d_{\min}, \mbox{ \;\;\; for all $i \neq j$.}$$
\end{defi}
In general, the packing set is not explicitly constructed, but its existence is guaranteed via combinatorial arguments. In Remarks \ref{rem:binyuset} and \ref{rem:castronowakset}
below,  existence of packing sets  is guaranteed by applying the Gilbert--Varshamov bound. In Remark \ref{rem:candesset}, the existence of a packing set is guaranteed via the probabilistic method.
We  emphasize that we use these existing packing set constructions: our contribution is to provide tighter lower bounds than can be obtained using Fano's inequality. It is possible that the resulting risk lower bounds could be improved by a further constant factor, by varying the packing set construction.

In statistical language, we  think of the packing set $\pack$ as multiple hypotheses to be distinguished on the basis of  data. An alternative information-theoretic interpretation is to think of $\pack$ as a codebook, that is a collection of $M$ codewords, one of which is transmitted over a noisy communication channel.
Given a packing set $\pack$, any estimator $\whth$ provides a way to distinguish between multiple hypotheses (act as a channel decoder) as follows: given $\whth$, we choose 
$\whi = \argmin_j d(\whth, \theta_j)$, i.e. the index of the closest value in the packing set. In coding theory, this is called the minimum distance decoder.

 If $\theta_i$ is the true value, a simple triangle inequality argument shows that 
$\left\{ \whi \neq i \right\}  \Rightarrow \left\{ d( \theta_i, \whth) \geq d_{\min}/2 \right\}$.
Taking $d_{\min} = 2 A \psi_n$,   we can bound the bracketed term on the RHS of  \eqref{eq:tsyb} by the average error probability $\e_M$ of the optimal decoder $i^* = i^*(\vc{Y})$,
since
\begin{align}
  \sup_{\theta \in \FF} \,  \pr \left( d( \theta, \whth) \geq A \psi_n \right) & \geq \max_{i \in \{1, \ldots, M \} } \pr \left( d( \theta_i, \whth) \geq A \psi_n \right) \nonumber \\
& \geq \max_{i \in \{1, \ldots, M \} } \pr( \theta_{\, \whi} \neq \theta_i)   \nonumber \\
& \geq   \max_{i \in \{1, \ldots, M \} } \pr( \theta_{i^*} \neq \theta_i) \nonumber \\
& \geq  \frac{1}{M} \sum_{i =1 }^M \pr( \theta_{i^*} \neq \theta_i) =: \e_M. \label{eq:tocontrol0}
\end{align}
This calculation and argument are standard in the literature (see for example \cite[Eq. (2.9)]{tsybakov09Book}, \cite[Corollary 2.19]{massart2007}).
By substituting \eqref{eq:tocontrol0} in \eqref{eq:tsyb}, we deduce
\be \inf_{\whth} \sup_{\theta \in \FF} \ep \left[ w \left( \frac{1}{\psi_n} d(\theta, \whth) \right) \right]  \geq w(A) \, \e_M. \label{eq:riskbd}
\ee
 Our main focus is to obtain a sharp and easily computable bound for $\e_M$.
A standard technique,  dating back to  Ibragimov and Khasminskii \cite{ibragimov}, is to bound $\e_M$  using Fano's inequality, which gives the bound \cite[Lemma 2.10]{tsybakov09Book}
\be 
\e_M \geq 1 - \frac{ \log 2 + \frac{1}{M} \sum_{i=1}^M D(P_{\theta_i} \| \ol{P}) }{\log M},  \label{eq:Fano}
\ee
where $\ol{P} := \frac{1}{M} \sum_{i=1}^M P_{\theta_i}$, and $D(P\|Q)$ is the  Kullback--Leibler (KL) divergence. To apply \eqref{eq:Fano}, one typically obtains a bound of the form 
\[  \frac{1}{M} \sum_{i=1}^M D(P_{\theta_i} \| \ol{P}) \leq \alpha \log M,  \]
for some constant $\alpha \in (0,1)$  (see  \cite[Section 2.7.1]{tsybakov09Book}). Then \eqref{eq:Fano} implies that $ \e_M \geq 1- \alpha - \frac{\log 2}{ \log M}$, which
 converges to $(1-\alpha)$  for large sample sizes $n$ (assuming $\log M \to \infty$ as $n \to \infty$), meaning that we deduce a weak converse,  and \eqref{eq:tocontrol0} gives a lower bound on the risk (via \eqref{eq:tsyb}). 

 The remainder of the paper is organized as follows. In Section \ref{sec:main_lb}, we derive a  lower bound on $\e_M$ (Theorem \ref{thm:main}) that  strengthens Fano's inequality.
 In Section \ref{subsec:rel_work}, we  discuss related prior work.  We then apply Theorem \ref{thm:main} to three high-dimensional estimation problems,  in each case showing  the average error probability $\e_M \to 1$ as $n \to \infty$ (strong converse).   In each case, our method replaces the Fano-based part of the argument which gives a weak converse.  In   Section \ref{sec:densityest}, we give a strong converse for a density estimation problem studied by Yu  \cite{BinYu}. In Section \ref{sec:active}, we obtain strengthened risk lower bounds for active learning of a binary classifier, following Castro and Nowak  \cite{castroNowak08}.
In Section \ref{sec:cs_bound},  we use Theorem \ref{thm:main} to  improve (by a factor of nearly 8) lower bounds of  Cand{\`e}s and Davenport \cite{candesDav13} for the minimax mean-squared error in compressed sensing.
\section{Lower bound on the Average Error Probability} \label{sec:main_lb}
We  bound the average error probability $\e_M$ in \eqref{eq:tocontrol0}  using a different \emph{binary} hypothesis testing problem. 
Adopting the formalism of \cite{ppv10}, we consider a random variable $S$ representing a message chosen uniformly at random
from  $\{ 1, \ldots, M \}$. The message $S$ is acted on by the simple encoder that generates codeword $\theta = \theta_S$, giving an induced distribution $\pi_\theta$  uniform over the set $\{ \theta_1, \ldots, \theta_M \}$.

We think of $\vc{Y} = (Y_1, \ldots, Y_n)$ as the output of a channel with input $\theta$.  Using arguments from \cite{ppv10, vazquez16}, we bound the desired average error probability of the optimal decoder \eqref{eq:tocontrol0} in terms of the Type I error probability of the following binary hypothesis testing problem:
\begin{align}
H_0: (\theta, \vc{Y}) & \sim  Q_{\theta \vc{Y}} := \pi_\theta Q_\vc{Y} \label{eq:nullhyp} \\
H_1: (\theta, \vc{Y}) & \sim  P_{\theta \vc{Y}} := \pi_\theta P_{\vc{Y}|\theta}, \label{eq:althyp}
\end{align}
for some probability distribution $Q_{\vc{Y}}$  that does not depend on $\theta$.    
In other words, we wish to determine whether $\theta$ and $\vc{Y}$ are independent, or are generated by the true underlying channel model.  We assume that the measure
 $Q_{\vc{Y}}$ dominates $P_{\vc{Y}|\theta_i}$ for $1 \leq i \leq M$, and hence the Radon-Nikodym derivative $\frac{d P_{\vc{Y}|\theta_i} }{d 
 Q_{\vc{Y}} }$ exists.
 
The space of  $\vc{Y}$ is denoted by  
$\mc{Y}$. Throughout the paper, we use boldface notation to denote vectors of length $n$.

\begin{thm} \label{thm:main}
 Let $\e_M$  denote the average error probability of any decoder over channel $P_{\vc{Y}|\theta}$,
 for a channel code with input distribution $\pi_\theta$  uniform  over the
$M$ codewords $\{ \theta_1, \ldots, \theta_M \}$. 
For any $\lambda >0$, and any distribution  $Q_{\vc{Y}}$  over $\mc{Y}$  such that  $P_{\vc{Y}|\theta_i}$ is absolutely continuous with respect to $Q_{\vc{Y}}$ for $1 \leq i \leq M$,
\be \label{eq:main}
\e_M \geq  1 -  \frac{(1+\lambda)} { \left( \lambda M \right)^{\frac{\lambda}{1+\lambda}} } 
 \left[ \sum_{i =1}^M  \frac{1}{M}  \exp\left( \lambda D_{1+\lambda}(P_{\vc{Y}|\theta_i}  \| Q_{\vc{Y}}) \right)\right]^{ \frac{1}{1+\lambda}}.
\ee
Here $D_{1+\lambda}(P_{\vc{Y}|\theta_i}  \| Q_{\vc{Y}})$ is the  R\'{e}nyi divergence of order $(1+\lambda)$ defined as
\be
D_{1+\lambda}(P_{\vc{Y}|\theta_i}  \| Q_{\vc{Y}}) := \frac{1}{\lambda}
 \log \left(  \int_{\mc{Y}} \left( \frac{d P_{\vc{Y}|\theta_i} }{d Q_{\vc{Y}}} \right)^{1+\lambda} dQ_{\vc{Y}}  \right).
 \label{eq:Renyi_div_def}
 \ee
\end{thm}
The proof  uses the following lemma, itself proved in Appendix \ref{sec:proofch_hyp2}.
\begin{lemma} 
With the assumptions and notation of Theorem \ref{thm:main} we have
  for any $\gamma >0$
 \be  \frac{1}{M} \geq \frac{1}{\gamma} \left(1- \e_M -  P_{\theta\vc{Y}} \left[ \frac{d P_{\vc{Y}|\theta} }{d Q_{\vc{Y}}} > \gamma \right]  \right).  \label{eq:M_lb2} \ee
\label{lem:ch_hyp2}
\end{lemma}
%
\begin{proof}[Proof of Theorem \ref{thm:main}]
Writing $\II(\cdot)$ for the indicator function, the probability  in \eqref{eq:M_lb2} satisfies
\begin{align}
&  P_{\theta \vc{Y}} \left[\frac{d P_{\vc{Y}|\theta} }{d Q_{\vc{Y}}} > \gamma \right]    =     \int_{\mc{Y}} \sum_{i =1}^M \frac{1}{M} \,  \II( \theta = \theta_i) \, \frac{d P_{\vc{Y}|\theta} }{d Q_{\vc{Y}}} (\vc{y}) \,  \II \left[ \frac{d P_{\vc{Y}|\theta} }{d Q_{\vc{Y}}} (\vc{y}) > \gamma \right]  dQ_{\vc{Y}}(\vc{y}) \nonumber \\
& \quad  = \int_{\mc{Y}} \sum_{i =1}^M \frac{1}{M} \,   \frac{d P_{\vc{Y}|\theta_i} }{d Q_{\vc{Y}}} (\vc{y}) \, \II \left[ \frac{d P_{\vc{Y}|\theta_i} }{d Q_{\vc{Y}}} (\vc{y}) > \gamma \right]  dQ_{\vc{Y}}(\vc{y}) \nonumber  \\
& \quad \leq  \int_{\mc{Y}} \sum_{i =1}^M \frac{1}{M} \,    \frac{d P_{\vc{Y}|\theta_i} }{d Q_{\vc{Y}}} (\vc{y}) \,  \left(  \frac{1}{\gamma} 
\frac{d P_{\vc{Y}|\theta_i} }{d Q_{\vc{Y}}} (\vc{y}) \right)^{\lambda}  d Q_{\vc{Y}}(\vc{y}) \mbox{ \;\;\;\;\; for $\lambda > 0$.} \label{eq:markov}
\end{align}
Using this bound \eqref{eq:markov}  in Lemma \ref{lem:ch_hyp2}, we have
\be   
 \frac{1}{M} \geq \sup_{\gamma >0} \left[ \frac{1-\e_M}{\gamma}  - \frac{1}{\gamma^{1+\lambda}} \sum_{i =1}^M  \frac{1}{M}  \int_{\mc{Y}} \left( \frac{d P_{\vc{Y}|\theta_i} }{d Q_{\vc{Y}}} (\vc{y}) \right)^{1+\lambda} dQ_{\vc{Y}}(\vc{y}) \right].   
 \label{eq:sup_gam_bnd} \ee
Computing the maximum over $\gamma >0$ and rearranging, we get  \eqref{eq:main}. 
\end{proof}

\begin{rem}
As shown in the active learning example in Sec. \ref{sec:active}, one can use upper bounds  for the R\'{e}nyi divergence in \eqref{eq:Renyi_div_def}  to obtain lower bounds for $\e_M$. Such upper bounds can found, for example, in \cite{sasonV16,sasonV15upper}.
\end{rem}

\begin{rem}
In Appendix \ref{app:Fano}, we show how Fano's inequality  in \eqref{eq:Fano} can be obtained from the lower bound on Theorem \ref{thm:main}. Furthermore, the examples in the next three sections show that Theorem \ref{thm:main} yields strictly better lower bounds than the Fano-based approach.   
\end{rem}

\begin{rem}
 If we assume that each $P_{\vc{Y}|\theta_i}$ has a density $p_{\theta_i}(\vc{y})$ with respect to a common reference measure $\mu$, then the choice of $Q_{\vc{Y}}$ that maximizes the lower bound in \eqref{eq:main}  has the following density with respect to $\mu$\cite{sibson1969information,naki17augustin}:
\[  q^*(\vc{y}) = \frac{1}{C}\left( \sum_{i=1}^M \frac{1}{M} (p_{\theta_i}(\vc{y}))^{1+\lambda} \right)^{\frac{1}{1+\lambda}},  \]
where 
the normalizing constant $C =  \int_{\mc{Y}} \, \left( \sum_{i=1}^M \frac{1}{M} (p_{\theta_i}(\vc{y}))^{1+\lambda} \right)^{\frac{1}{1+\lambda}} d\mu(\vc{y})$.  However,  the bound in \eqref{eq:main} is generally not computable with this choice of $Q_{\vc{Y}}$.  As we will see in the following sections, the structure of the problem often suggests a natural choice for $Q_{\vc{Y}}$  that yields a computable lower bound. 
\label{rem:optQ}
\end{rem}

\subsection{Related work} \label{subsec:rel_work}

In \cite[Proposition 2.2]{tsybakov09Book}, Tsybakov gives a result similar to Lemma \ref{lem:ch_hyp2}. This result can then be  used to obtain a lower bound on $\e_M$ using the average pairwise $\chi^2$-distance between $q$ and the elements of the packing set \cite[Theorem 2.6]{tsybakov09Book}. This bound is similar to the one  obtained by using $\lambda=1$ in  Theorem \ref{thm:main}. In this paper, we show that via two examples (active learning and compressed sensing) that Theorem \ref{thm:main} can be applied with a general $\lambda >0$ to obtain stronger non-asymptotic bounds. Furthermore, as $n \to \infty$, Theorem \ref{thm:main} gives a strong converse ($\e_M \to 1$), unlike  Fano's inequality.

Birg\'{e}  \cite{birge2005} gives stronger, but less transparent, bounds than Fano's inequality using Fano-type arguments; again, $\e_M$ is bounded in terms of an average of Kullback--Leibler divergences, but these are used as the argument for a function, rather than directly substituted. Sason and Verd\'u \cite[Section 3]{sasonV18arimoto} recently derived a generalized Fano's inequality in terms of the Arimoto-R\'enyi conditional entropy. They also obtained upper bounds on $\e_M$ in terms of the pairwise R\'enyi divergences \cite[Section 4]{sasonV18arimoto}.

Note that an alternative approach to hypothesis testing bounds, avoiding the use of  Fano's inequality, is given by Assouad \cite{assouad1983}. Indeed,  \cite{BinYu} makes a detailed comparison between Fano-based bounds and those coming from Assouad's Lemma \cite{assouad1983}, finding little practical difference. Indeed \cite{BinYu} quotes  Birg\'{e} \cite[p. 279]{birge1986}: ``[Fano] is in a sense more general because it applies in more general situations. It could also replace Assouad's Lemma in almost any practical case \ldots ''.

Using Fano's inequality, Yang and Barron \cite{YangBarron99}  obtained order-optimal minimax risk lower bounds  that depend  only on global metric entropy features of the underlying function class, without explicitly constructing a packing set. The required metric entropy features (bounds on the packing number and covering number) are available from results in approximation theory for many function classes of interest. 
Guntuboyina \cite{guntu11}  obtained a lower bound on the average error probability in terms of general $f$-divergences, and also generalized the metric entropy results of Yang and Barron
\cite{YangBarron99} to certain $f$-divergences such as the $\chi^2$-divergence.
An interesting direction for future work would be to obtain a  non-asymptotic result analogous to Theorem \ref{thm:main}  for the case where only the global metric entropy features are available.

An important historical remark is that Hayashi and Nagaoka \cite{hayashi2} first linked channel coding and binary hypothesis testing, 
 with later work \cite{hayashi3} by Hayashi 
clarifying this approach and Nagaoka \cite{nagaoka} using similar ideas to derive strong converse results. The recent work by Vazquez-Vilar et al.  \cite{vazquez16} also provides  results characterizing the average error probability of channel coding in terms of the Type I error of a binary hypothesis test.
This link with channel coding has been used in other contexts to prove strong converse results, 
including \cite{johnson-ppv}, which derived strong converse results for the  group testing problem.

\section{Application to density estimation} \label{sec:densityest}
For the remainder of this paper, we show how Theorem \ref{thm:main} can be applied to a number of high-dimensional estimation problems. In this section
we apply Theorem \ref{thm:main} to the following  density estimation problem taken from  Yu \cite[Example 2, P.431]{BinYu}. 
Let $\mc{F}$ be the class of smooth densities on  $[0,1]$ such that for any density $\theta \in \mc{F}$, we have
\begin{align*}
& \int_{0}^1 \theta(x) dx =1, \qquad  a_0 \leq \theta(x) \leq a_1 < \infty, \qquad  \abs{\theta^{\prime \prime}(x)} \leq a_2, \quad x \in \mathbb{R},
\end{align*}
for some positive constants $a_0, a_1, a_2$.  The goal is to estimate the density $\theta$ from $\vc{Y} = (Y_1, \ldots, Y_n)$, where
$\{ Y_i \}$ are generated  IID from $\theta$.  We want to bound from below the risk of any estimator $\whth_n = \whth_n(\vc{Y})$, where the loss is measured using squared Hellinger distance, i.e., 
\begin{equation} \label{eq:hell}
 d(\theta, \whth_n) = \int_{0}^1 \left( \sqrt{\theta(x)} - \sqrt{\whth_n(x)} \right)^2 dx.
\end{equation}

The packing set in \cite{BinYu} is constructed via a hypercube class of densities defined via small perturbations of the uniform density on
$[0,1]$.
Fix a smooth, bounded function $g(x)$ with
\begin{equation}
\int_0^1 g(x) dx  =  0 \mbox{ \;\;\; and \;\;\;}
\int_0^1 \left( g(x) \right)^2 dx  =  a.  \label{eq:gprop2} 
\end{equation}
We partition the unit interval $[0,1]$ into $m$  subintervals of length $1/m$,
and perturb the uniform density on each subinterval by a small amount, proportional to a version of $g$
rescaled and translated to lie on that subinterval. That is, for some sufficiently small fixed constant $c$,
 we can define the functions
\begin{equation} \label{eq:gjdef}
 g_j(x) = \frac{c}{m^2} g(m x - j) \II \left( \frac{j}{m} \leq x < \frac{j+1}{m} \right),
\mbox{\;\;\; for $j =0, \ldots, m-1$ }.\end{equation}
 Considering perturbations of the uniform density by $\pm \{ g_j \}$, define  the following hypercube class of joint densities indexed by $\vd{\tau} = 
\left( \tau_1, \ldots, \tau_m \right) \in \{ \pm 1 \}^m$:
\begin{equation}
\MM_m 
= \left\{ 
\ftau(y) = 1 + \sum_{j=0}^{m-1} \tau_j g_j(y) 
\right\}.
\label{eq:ftau_def}
\end{equation}
The bandwidth parameter $m$ will be chosen later as an increasing function of $n$, to optimize the risk lower bound.

\begin{rem}[Packing set construction] \cite[Lemma 4]{BinYu}  \label{rem:binyuset}
There exists  a subset $\AAA \subseteq \{ -1, 1 \}^m$ of size $M \geq \exp(c_0 m)$, where $c_0 \simeq 0.082$, whose elements have minimum pairwise
Hamming distance at least $m/3$. It is then shown in \cite{BinYu} that this results in a packing set  of densities 
$$\packval{\frac{ac^2}{3m^4}}
= \{ \ftau: \vd{\tau} \in \AAA \} \subseteq \MM_m,$$
where $ac^2/(3m^4)$
is a lower bound on the squared Hellinger distance (see \eqref{eq:hell}) between distinct densities  in the packing set.  Here $a$ is defined in \eqref{eq:gprop2},  and $c$ is defined in \eqref{eq:gjdef}. We use exactly this packing set $\packval{ac^2/(3m^4)}$ as the set of $M$ codewords
$\{ \theta_1, \ldots, \theta_M \}$ in Theorem \ref{thm:main}.
\end{rem}

We now use Theorem \ref{thm:main} to bound the risk. To do this, 
we first state an explicit bound (to be proved in Appendix \ref{sec:prooflemexmplicit}) on the bracketed term in \eqref{eq:main}, for $\lambda =1$.

\begin{lemma} \label{lem:explicit}  Taking $Q_{\vc{Y}}$ to be the uniform measure on $[0,1]^n$ and identifying each $\frac{dP_{\vc{Y}|\theta_i}}{d Q_{\vc{Y}}}$ with a density $\ftau^n(\vc{y}) =  \prod_{i=1}^n f_{\vd{\tau}}(y_i)$ for $\tau \in \AAA$,
with $\lambda=1$, the bracketed term in \eqref{eq:main} becomes:
\begin{align*}
 \left[ \sum_{i =1}^M  \frac{1}{M} 
 \int_{\mc{Y}} \left( \frac{dP_{\vc{Y}|\theta_i} }{dQ_{\vc{Y}}}(\vc{y}) \right)^{2} dQ_{\vc{Y}}(\vc{y}) \right]^{ \frac{1}{2}}
&  = \left[ \sum_{\vd{\tau} \in \AAA} \frac{1}{M} \int_{[0,1]^n}  \ftau^n(\vc{y})^2  d\vc{y} \right]^{\frac{1}{2}} \\
& \leq 
 \exp \left( \frac{ c^2 a n}{2 m^4} \right) .
 \end{align*}
\end{lemma}

Combining Lemma \ref{lem:explicit}  with Theorem \ref{thm:main},  we deduce the following lower bound.
\begin{prop} \label{prop:binyumain}
For any positive constant  $\nu < (c_0/(c^2 a))^{1/5}$,  the risk of any estimator $\whth_n$ satisfies
\begin{equation} \label{eq:riskbd2}
 \sup_{\theta \in \FF} \expec d(\whth_n, \theta)    \geq  \frac{c^2 a \nu^4 }{6} \,  n^{-4/5}\, \e_M,
\end{equation}
where
\begin{equation}
\e_M \geq
1- 2 \exp \left( - \frac{n^{1/5}}{2 \nu} \left( c_0 - \nu^5 c^2 a \right) \right). \label{eq:n5bound}
\end{equation}
Therefore, for large $n$ we have
\be
 \sup_{\theta \in \FF} \expec d(\whth_n, \theta)    \geq   \frac{c_0^{4/5} (c^2 a)^{1/5}}{6} n^{-4/5}(1-o(1)).
 \label{eq:den_bnd_asymp}
\ee
\end{prop}
\begin{proof} We apply Theorem \ref{thm:main} by  setting the minimum distance $ac^2/(3m^4)$  of the packing set  in Remark \ref{rem:binyuset} to  $2 A \psi_n$. Taking $A=1$, we obtain  $\psi_n = c^2 a /(6 m^{4})$.  
Taking $w$ to be the identity in \eqref{eq:riskbd}, we deduce
\[ \max_j \ \expec d(\hat{\theta}, \theta_j)  \geq \psi_n \e_M   = \frac{c^2 a }{6 m^{4}} \, \e_M.  \]
Taking $\lambda=1$ in Theorem \ref{thm:main} and using Lemma \ref{lem:explicit}, we  bound $\epsilon_M$ as
\begin{align*}
\epsilon_M & \geq  1- \frac{2}{\sqrt{M}} \exp \left( \frac{ c^2 a n}{2 m^4} \right) \stackrel{(a)}{\geq}  1 - 2  \exp \left( - \frac{ c_0 m}{2} \right)  \exp \left( \frac{ c^2 a n}{2 m^4} \right),
\end{align*}
where $(a)$ is obtained using the fact that the packing set of densities $\packval{ac^2/(3m^4)}$ has size $M \geq \exp(c_0 m)$, as described in Remark \ref{rem:binyuset} above. We therefore have 
\be \label{eq:binyum_bnd}
\max_j \ \expec d(\hat{\theta}, \theta_j)  \geq \frac{c^2 a }{6 m^{4}} \left[ 1 - 2  \exp \left( - \frac{m}{2} \left( c_0  -  \frac{ c^2 a n}{ m^5} \right) \right)  \right].
\ee
The result \eqref{eq:riskbd2} follows by taking $m = n^{1/5}/\nu$.

To obtain the asymptotic bound in \eqref{eq:den_bnd_asymp}, we take  $\nu$ to approach $ (c_0/(c^2 a))^{1/5}$ as  $n \rightarrow
\infty$, but slowly enough that ensure that the exponent on  the RHS of \eqref{eq:n5bound} is negative so that $\e_M$ tends to 1. For example, we can take $ \nu = \left(\frac{c_0}{c^2 a} \right)^{1/5}(1- n^{-1/\kappa})$ for  $\kappa >25$. 
\end{proof}

\begin{rem}
 The paper \cite{BinYu}  derives Fano-type bounds in this setting: combining Lemmas 3 and 5 of \cite{BinYu} and taking taking $m = n^{1/5}/\nu$ gives the same bound as \eqref{eq:riskbd2}, but with a looser lower bound on $\e_M$ given by 
 \begin{align}
 \e_M &  \geq \left( 1 - \frac{1}{c_0}\left( \frac{2  c^2 a  \nu^5}{(1- c_g)} + \frac{\log 2}{n^{1/5}} \right)   \right).
\label{eq:binyufano}
\end{align}
 For the  bound \eqref{eq:binyufano} to be meaningful, we need $\nu < \left( \frac{c_0}{c^2 a} \frac{1-c_g}{2} \right)^{1/5}$, where $c_g = c \sup_x \abs{g(x)}$. The scaling factor $c$ has to be chosen so that $c_g <1$.  

Thus  Proposition \ref{prop:binyumain} provides a strong converse
(error probability tending to $1$), whereas the result \eqref{eq:binyufano} extracted from \cite{BinYu} gives a weak converse
(error probability bounded away from $0$).   Our bound also offers greater flexibility in choosing $\nu$ and removes the need to control $c_g$.
\end{rem}
\begin{rem}
Theorem \ref{thm:main} can similarly be applied to obtain strong converses for estimating densities belonging to either H{\"o}lder or Sobolev classes, strengthening the risk lower bounds described in \cite[Sec. 2.6.1]{tsybakov09Book}
\end{rem}

\section{Application to active learning of a classifier} \label{sec:active}

In this section, we use Theorem \ref{thm:main} to derive strengthened minimax lower bounds for active learning algorithms for a family of classification problems introduced by
Castro and Nowak \cite{castroNowak08} (see also \cite{tsybakov04}). We use the explicit packing set construction of \cite{castroNowak08},
but modify their notation for consistency.

Consider data of the form $\vc{Y} = (\vc{U}, \vc{V}) = \left( (U_1, V_1), \ldots, (U_n, V_n) \right)$. Each pair $(U_r, V_r)$ consists of a 
feature vector $U_r \in \re^d$ (where we assume $d \geq 2$) and a binary label $V_r \in \{ 0, 1 \}$, and is drawn independently from an underlying joint distribution
$P_{UV} = P_U P_{V|U}$.  The goal of classification is to predict the value of label $V$, given a future $U$ observation. This is done via  $G$, a measurable subset of $\re^d$. Given a $U \in \mathbb{R}^d$, the classifier estimates its label as $\whv := \whv(U) := \II( U \in G)$. The risk of a classifier is the probability of classification error, given by
$$ R(G) = \pr( \whv \neq V) = \pr( \II(U \in G) \neq V),$$
where $(U,V) \sim P_{UV}$.
It is well-known (see \cite{tsybakov04}) that, given knowledge of $P_{UV}$, the Bayes-optimal classifier is  
$$ G^* = \{u \in \mathbb{R}^d: \eta(u) \geq 1/2 \}, $$
where the feature conditional probability $\eta(u) = P_{V|U}(1|u)$.
As  $P_{UV}$ is unknown, our goal is to estimate $G^*$ from data $\vc{Y}$. 
The performance of classifier $\whG_n$ is measured by  \emph{excess risk}  (or regret) \cite[Eq. (1)]{castroNowak08}
$$R(\whG_n)  - R(G^*) = \int_{\whG_n \Delta G^*} \abs{2 \eta(u) - 1} dP_U(u),$$
where $\Delta$ represents the symmetric difference between sets.
For the remainder of this section, as in \cite{castroNowak08}, we will assume that $P_U$ is supported on $[0,1]^d$. 
It is clear that the difficulty of a classification problem will depend on both the shape of the  boundary of $G^*$ and the behaviour of $(2\eta(u) - 1)$ for $u$ close to this boundary. We consider the class of joint distribution functions $\mathsf{BF}(\alpha, \kappa, L, c)$,
defined formally in \cite[Section IV]{castroNowak08}. For our purposes it suffices to understand this class as a set of distributions $P_{UV}$ such that: 
\begin{enumerate}
\item The boundary of $G^*$ can be expressed as an $\alpha$-H\"{o}lder smooth function with constant $L$.
\item The value of $|\eta(u) - 1/2| $ is at least $c D^{\kappa - 1}$ for points $u$ at distance $D$ from the boundary, where $\kappa \geq 1$.
\end{enumerate}

Algorithms that  attempt to learn the Bayes-optimal classifier $G^*$ from data are categorized as passive or active. \emph{Passive} learning algorithms aim to learn $G^*$ from a pre-specified, possibly random, choice of $(U_1, \ldots, U_n)$ and the corresponding labels $(V_1, \ldots, V_n)$.
In contrast, \emph{active} learning algorithms  choose each $U_r$ based on previous values $(U,V)_r^- := \left( U_1, \ldots, U_{r-1}, V_1, \ldots, V_{r-1} \right)$.
This allows us to adaptively probe the boundary of $G^*$.  
A (randomized) active learning algorithm is defined by a sequence of conditional distributions  $P_{U_r | (U,V)_r^-} := P_{U_r | \left( U_1, \ldots, U_{r-1}, V_1, \ldots, V_{r-1} \right)}$, which defines the joint distribution as follows:
\begin{equation}  P_{\vc{U} \vc{V}} 
:= \prod_{r=1}^n P_{U_r | (U,V)_r^-}  P_{V_r|U_r} 
\end{equation}
where  $P_{V_r|U_r}  \equiv P_{V|U}$; in particular, conditioned on $U_r$, label $V_r$ is  independent of $(U_1, \ldots, U_{r-1})$. We assume that  for each $r$, the conditional distribution
$P_{U_r | (U,V)_r^-}$ has a density $p_{U_r | (U,V)_r^-}$ with respect to Lebesgue measure on $[0,1]^d$. Note that active learning algorithms correspond to channel coding with feedback, and to adaptive group testing algorithms  \cite{johnson-ppv}. Passive learning corresponds to channel coding without feedback, and to non-adaptive group testing algorithms.

We provide lower bounds on the excess risk of active learning algorithms that strengthen those in \cite[Theorem 3]{castroNowak08}, but our techniques can also be applied to  \cite[Theorem 4]{castroNowak08}, which applies in the passive case. We use the packing set  constructed in \cite{castroNowak08},  which is defined via a hypercube class of joint distributions on $(U,V)$. Fix an integer $m$ (to be chosen later as a function of $n$).
For each vector $\vd{\tau} \in  \{0,1 \}^{m^{d-1}}$, Castro and Nowak \cite[Appendix C]{castroNowak08} construct a unique distribution of $(U,V)$ whose feature conditional probability
is denoted by $\eta_{\vd{\tau}}(u)$, and the corresponding Bayes classifier  is denoted by $G^*_{\vd{\tau}}$. We denote this hypercube class of  $2^{m^{d-1}}$  distributions by $\FF_m$. Each distribution in $\FF_m$ has the same $U$-marginal $P_U$. Thus the joint distribution is determined by the  conditional distribution $P_{V|U}$. The conditional distributions  in $\FF_m$ (equivalently, the feature conditional probabilities $\eta_{\vd{\tau}}(u)$ for each $\vd{\tau} \in  \{0,1 \}^{m^{d-1}}$) are not explicitly defined here, but the definition can be found in the displayed equation at the foot of \cite[p.2350]{castroNowak08}.  The definition ensures that the hypercube class $\FF_m \subseteq \mathsf{BF}(\alpha, \kappa, L, c)$. 

The packing set defined in \cite[Appendix C]{castroNowak08} is a subset of distributions in $\FF_m$.
\begin{rem}[Packing set construction] \cite[Lemma 2]{castroNowak08} \label{rem:castronowakset}
There exists a subset $\AAA \subseteq \{0,1 \}^{m^{d-1}}$ of size $M +1$ with $ M \geq 2^{m^{d-1}/8}$, whose elements have minimum pairwise Hamming distance at least $m^{d-1}/8$. It is then shown in \cite{castroNowak08} that this results in a packing set of  functions
 $\packvalx{\beta_m/8} = \{ \eta_{\vd{\tau}}: \vd{\tau} \in \AAA \}$, where $\beta_m=  L H m^{-\alpha}$, and  $\eta_{\vd{\tau}}(1|u) = P_{V|U}(u)$. (Hence $1-\eta_{\vd{\tau}}(u) = P_{V|U}(0|u)$.)
Here $\beta_m/8$ is a lower bound on the set distance between distinct elements of $\packvalx{\beta_m/8}$,
defined as 
\be d_\Delta( G^*_{\vd{\tau}}, G^*_{\vd{\tau'}}) =
\int \II \left( u \in  \left( G^*_{\vd{\tau}} \Delta G^*_{\vd{\tau}'} \right) \right) du, \label{eq:set_dist_def} \ee
and  $H = \norm{h}_1$ is the norm of a suitable smooth function $h$.

Furthermore,  $\packvalx{\beta_m/8}$ contains the function $\eta_{\vd{0}}$,  corresponding to the point ${\vd{\tau}} = (0, 0, \ldots, 0)$ in the hypercube. We use the other $M$ functions in the packing set $\packvalx{\beta_m/8}$ to act as the $M$ codewords $\{ \theta_1, \ldots, \theta_M \}$ in Theorem \ref{thm:main}. 
The Bayes classifiers corresponding to these codewords are denoted by $G_1^*, \ldots, G_M^*$.
\end{rem}

As in Section \ref{sec:densityest}, we prove an explicit bound on the bracketed term \eqref{eq:main} in Theorem \ref{thm:main}, with  
$(\vc{u}, \vc{v})$ corresponding to $\vc{y}$ in  \eqref{eq:main}.

\begin{lemma} \label{lem:classmain} 
For an  active learning algorithm described by $\prod_{r=1}^n P(U_r | (U,V)_r^-)$, we take
 $Q_{\vc{U},\vd{V}}(\vd{U}, \vd{V}) := \prod_{r=1}^n P(U_r | (U,V)_r^-) \prod_{r=1}^n P_{\vc{0}}( V_r |U_r)$, where $P_{\vc{0}}( V_r |U_r)$ is the conditional probability mass function
 determined by $\eta_{\vc{0}}$ which 
 corresponds to the point ${\vd{\tau}} = (0, 0, \ldots, 0)$ in the hypercube.
 
   Further, for each $\vd{\tau} \in \AAA $ and $\vd{\tau} \neq \vd{0}$,   we can take
    \[ P_{\vd{\tau}}(\vc{U},\vc{V}) := \prod_{r=1}^n P(U_r | (U,V)_r^-)  \prod_{r=1}^n P_{\vd{\tau}}( V_r |U_r),\] where $P_{\vd{\tau}}(V_r|U_r)$ is the conditional probability mass function determined by $\eta_{\vd{\tau}}$. Then,   for any $\lambda \in (0,1]$,  the bracketed term in \eqref{eq:main} satisfies
\begin{equation}
\sum_{\vd{\tau} \in \AAA, \vd{\tau} \neq \vd{0}} \frac{1}{M} \int_{\mc{Y}_n }  
\left( \frac{  dP_{\vd{\tau} } }{ dQ_{\vc{U},\vc{V}}} (\vc{u}, \vc{v}) \right)^{1+\lambda} dQ_{\vc{U},\vc{V}}(\vc{u}, \vc{v})  
\leq    \exp \left( \frac{16  c^2  \beta_m^{2(\kappa-1)} \lambda n }{(1 - 2c \beta_m) } \right). \label{eq:classmain} 
 \end{equation}
where for brevity we write an integral to represent integration and summation over the product space
$\mc{Y}_n = [0,1]^{d \times n} \otimes \{ 0, 1 \}^n$, and $\beta_m=LHm^{-\alpha}$ as defined in Remark \ref{rem:castronowakset}.
 \end{lemma}
\begin{proof} See Appendix \ref{sec:proofclassmain}. \end{proof}
Combining Lemma \ref{lem:classmain}  with Theorem \ref{thm:main},  we deduce the following lower bound. 
\begin{prop} \label{prop:activemain}
Let $\rho= (d-1)/\alpha$. For any positive constant  $\nu$, the risk  of a classifier $\whG_n$ learnt via any active learning algorithm satisfies
\begin{equation} \label{eq:riskbd2A}
\sup_{P_{UV} \in \mathsf{BF}(\alpha, \kappa, L, c)}\, \left\{ \ep [R(\whG_n) ] - R(G^*) \right\}   \geq  \frac{4c \nu^{\kappa \alpha}}{\kappa} \left(\frac{LH}{32} \right)^{\kappa} n^{-\frac{\kappa}{2\kappa - 2 + \rho}} \, \e_M,
\end{equation}
where  
\begin{equation}
\e_M \geq
1 -  \frac{1+\lambda}{\lambda^{\lambda/(1+\lambda)}}  \exp \left( \frac{ - \lambda n^{\frac{\rho}{2\kappa -2 + \rho}}   }{(1+\lambda) \nu^{d-1}} \left( \frac{\log 2}{8}  
- \frac{16 c^2   (LH)^{2\kappa-2} \nu^{d-1+ 2\alpha(\kappa -1)} }{1- 2c LH n^{-1/(2\kappa-2 + \rho)}/\nu}  \right) \right)  
\label{eq:n5boundA}
\end{equation}
for any  $\lambda \in (0,1]$. Therefore, for large $n$ we have
\be
\begin{split}
& \sup_{P_{UV} \in \mathsf{BF}(\alpha, \kappa, L, c)}\, \left\{ \ep [R(\whG_n) ] - R(G^*) \right\} \\
&    \geq \left[ \frac{4 c}{ \kappa 32^{\kappa}} 
\left(\frac{\log 2}{128 c^2} \right)^{\frac{\kappa}{2\kappa- 2 + \rho}} \left( LH \right)^{\frac{\kappa \rho}{2\kappa-2 + \rho }} \right] n^{-\frac{\kappa}{2\kappa - 2 + \rho}}(1-o(1)).
\end{split}
\label{eq:learning_asymp}
\ee
\end{prop}
\begin{proof}
Consider the $M$ codewords chosen from the packing set $\packvalx{\beta_m/8}$, as described in Remark \ref{rem:castronowakset},  with corresponding Bayes classifiers $G^*_1, \ldots, G^*_M$ . The minimum pairwise set distance between these Bayes classifiers is at least $\beta_m/8$.
Equating the minimum distance of the packing set given by $2 A \psi_n$ (in Theorem \ref{thm:main}) to $\beta_m/8$,  taking $A=1$ we obtain
$\psi_n = \beta_m/16 = LH m^{-\alpha}/16$.  

Using Lemma \ref{lem:classmain} in Theorem \ref{thm:main},  for any $\lambda \in (0,1]$ the average error probability $\epsilon_M$ can be bounded from below as
\begin{align}
\label{eq:em_bndclass}
 \epsilon_M & \geq  1- \frac{1+\lambda}{\lambda^{\lambda/(1+\lambda)}} M^{-\lambda/(1+\lambda)}   \exp \left( \frac{16  c^2  \beta_m^{2(\kappa-1)} \lambda n}{(1 - 2c \beta_m)(1+\lambda) } \right)  \\
  & \stackrel{(a)}{\geq}  1 -  \frac{1+\lambda}{\lambda^{\lambda/(1+\lambda)}}
    \exp \left( \frac{\lambda}{1+\lambda} \left( \frac{16  c^2 \beta_m^{2(\kappa-1)}n }{(1 - 2c \beta_m)} - \frac{m^{d-1} \log 2}{8} \right) \right).
\end{align}
where inequality $(a)$ is obtained using the fact that packing set of distributions $\packvalx{\beta_m/8}$ has $M \geq 2^{m^{d-1}/8}$, as described in Remark \ref{rem:castronowakset} above. 

Now, consider any  distribution $P_{UV} \in \mathsf{BF}(\alpha, \kappa, L, c)$ with Bayes classifier $G^*$. It is shown in \cite[p.2351]{castroNowak08} that the event 
\[  \left\{  d_\Delta( \whG_n, G^*) \geq  \psi_n \right\}  \Rightarrow \left\{ R(\whG_n) - R(G^*) \geq \min\left(  \frac{4c}{\kappa 2^{\kappa}} \psi^\kappa_n, \psi_n \right) \right\}. \]
Defining $ f(\psi_n) := \min\left(  \frac{4c}{\kappa 2^{\kappa}} \psi^\kappa_n, \psi_n \right)$, we therefore obtain the following chain of inequalities:
\begin{align*}
\e_M & \stackrel{(b)}{\leq} \sup_{P_{UV} \in \mathsf{BF}(\alpha, \kappa, L, c)} \,  \pr \left( d_\Delta( \whG_n, G^*) \geq  \psi_n \right)  \\
&  \leq \sup_{P_{UV} \in \mathsf{BF}(\alpha, \kappa, L, c)} \,  \pr \left( R(\whG_n) - R(G^*) \geq f(\psi_n) \right)   \\
&   \leq  \sup_{P_{UV} \in \mathsf{BF}(\alpha, \kappa, L, c)} \,  \frac{  \expec[ R(\whG_n) ] - R(G^*)}{ f(\psi_n)},
\end{align*}
where inequality $(b)$ follows from \eqref{eq:tocontrol0}. Hence, using $\psi_n= LH m^{-\alpha}/16$, we have\footnote{As $m \gg 1$ and $\kappa \geq 1$, we assume for brevity that $f(\psi_n) =  \frac{4c}{\kappa 2^{\kappa}} \psi^\kappa_n$. This is always true for sufficiently large $m$ when $\kappa > 1$. However, if $\kappa=1$ and $c > \frac{1}{2}$, then $f(\psi_n) = \psi_n$; however, the definition of $c$ in \cite[Eq. (9)]{castroNowak08} implies that $c$ can be restricted to $(0,\frac{1}{2}]$ without loss of generality.}
\ben
\begin{split}
& \sup_{P_{UV} \in \mathsf{BF}(\alpha, \kappa, L, c)}  \, \expec[ R(\whG_n) ] - R(G^*)  \\
 &  \geq  f(\psi_n) \e_M   \\
 & = \frac{4c}{\kappa} \left(\frac{L H m^{-\alpha}}{32}\right)^{\kappa} \e_M  \\
& \geq  \frac{4c}{\kappa} \left(\frac{L H m^{-\alpha}}{32}\right)^{\kappa}  \left[ 1 -  \frac{1+\lambda}{\lambda^{\lambda/(1+\lambda)}}
    \exp \left( \frac{\lambda}{1+\lambda} \left( \frac{16  c^2 n \beta_m^{2(\kappa-1)} }{(1 - 2c \beta_m)} - \frac{m^{d-1} \log 2}{8} \right) \right) \right],
\end{split} 
\een
where the last inequality is obtained using \eqref{eq:em_bndclass}. The result follows by taking $m = n^{\frac{1}{\alpha(2\kappa -2) + d-1}}/\nu$.

To obtain \eqref{eq:learning_asymp}, we choose the supremum of $\nu$ such that $\e_M \to 1$ as $n \to \infty$ in order to obtain the largest possible prefactor in \eqref{eq:riskbd}.
\end{proof}

\begin{rem}
 The paper \cite{castroNowak08}  derives Fano-type bounds in this setting:  in particular, taking $m = n^{\frac{1}{\alpha(2\kappa -2) + d-1}}/\nu$, the computation in p.2351 of \cite{castroNowak08} together with Theorem 6 of that paper gives the same bound as \eqref{eq:riskbd2A}, but with a looser lower bound on $\e_M$ given by 
 \begin{align}
 \e_M &  \geq \left( 1 -  2 \xi - \sqrt{\frac{32 \xi \nu^{d-1}}{\log 2}} n^{-\frac{\rho}{4(\kappa-1) +2\rho }} \right), 
\label{eq:castrofano}
\end{align}
where $\xi = \frac{256}{\log 2} c^2 (LH)^{2\kappa -2} \nu$.
For the  bound \eqref{eq:castrofano} to be meaningful, we need $\xi < \frac{1}{2}$, which implies $\nu <  \frac{ \log 2}{512 c^2 (LH)^{2\kappa -2}}$. Again,  Proposition \ref{prop:activemain} provides a strong converse, while \eqref{eq:castrofano} provides a weak one ($\e_M$ bounded away from zero).
\end{rem}

\section{Application to compressed sensing} \label{sec:cs_bound}

We now describe how Theorem \ref{thm:main} can give improved risk lower bounds in compressed sensing. The goal in compressed sensing is to estimate a sparse vector $\mbf{x} \in \mathbb{R}^n$ from a measurement $\mbf{y} \in \mathbb{R}^m$ of the form
\be \vc{y} =\mbf{A} \vc{x} + \vc{w}. \label{eq:cs_model} \ee
Here $\mbf{A} \in \mathbb{R}^{m \times n}$ is  the (known) measurement matrix,   and $\vc{w} \sim \mc{N}(0, \sigma^2 \mathbf{I}_m)$ is the  noise vector.  Throughout this section $\norm{\vc{x}}$ denotes the $L_2$ Euclidean norm of a vector $\vc{x}$, and $\norm{\mbf{A}}_F$ denotes the Frobenius norm of a matrix $\mbf{A}$, defined by $\norm{\mbf{A}}_F^2 =
\trace \left( \mathbf{A^T} \mathbf{A} \right)$.
We assume that the signal $\vc{x}$ is $k$-sparse, by considering $\vc{x} \in \Sigma_k$, where 
\[ \Sigma_k := \left\{ \vc{x} \in \mathbb{R}^n:
\|  \vc{x} \|_0 \leq k, \norm{ \vc{x}} = 1 \right\}. \]

In the pioneering works \cite{CandesTaoLP, donohoCS, candesRT06} of Cand\`{e}s, Donoho, Romberg, and Tao, among others, it was shown that under suitable assumptions on $\vc{A}$ and the sparsity level $k$, the signal could be efficiently estimated to a high degree of accuracy, even when  
$m \ll n$. For example, when $\mbf{A}$ satisfies the Restricted Isometry Property \cite{candes08RIP}, reconstruction techniques based on minimizing the $L_1$ norm produce an estimate $\whx$  which satisfies
\[ \frac{1}{n}\norm{\vc{x} - \whx}^2  \leq C_0 \frac{k \sigma^2}{m} \log n \]
with high probability, provided that $m$ is of order at least $k \log (n/k)$ \cite{BickelRT09}. ($C_0$ is a universal positive constant.)   

To complement these achievability results, several authors, e.g.,  \cite{aeronSZ10,yeZ10,raskuttiWY11,candesDav13} have derived  lower bounds on the minimax risk under various assumptions on $\vc{A}$ and  $\vc{x}$. The minimax risk is defined as 
\begin{equation} \label{eq:minimaxcs}
\minimaxcs := \inf_{ \whx} \sup_{\vc{x} \in \Sigma_k} \ep \left[ \frac{1}{n} \norm{  \whx(\vc{y}) - \vc{x}}^2 \right],
\end{equation}
We show how Theorem \ref{thm:main} can be used to obtain a strong converse, improving by a constant factor the lower bound on $\minimaxcs$  obtained using Fano's inequality  by Cand\`{e}s and Davenport in \cite{candesDav13}. Using the probabilistic method, \cite{candesDav13} shows  the existence of a packing set of well-separated  vectors in $\Sigma_k$. To be specific:

\begin{rem}[Packing set construction] \cite[Lemma 2]{candesDav13}  \label{rem:candesset} There exists a subset $\mc{X} \subseteq \Sigma_k$ of size $M:= \abs{\mc{X}} = (n/k)^{k/4}$ whose elements $\vc{u}_i$ satisfy
\begin{enumerate}
\item $\norm{\vc{u}_i}^2  = 1$.
\item $\norm{ \vc{u}_i - \vc{u}_j}^2 \geq \frac{1}{2}$ for all $1\leq i,j \leq M$ such that $ i \neq j$.
\item \label{it:packprop} $\norm{\frac{1}{M}\sum_{i=1}^M \vc{u}_i \vc{u}_i^T - \frac{1}{n} \mbf{I} }_{\rm{op}} \leq \beta/n$.
 Here $\beta$ is a constant that can be made arbitrarily small with growing $n$.
\end{enumerate}
The set $\mc{X}$ gives a packing set $\packval{C/\sqrt{2}} := \{ \theta_1, \ldots, \theta_M \}$ of codewords with minimum distance
$\norm{ \theta_i - \theta_j} \geq \frac{C}{\sqrt{2}}$, simply by taking $\theta_i = C \vc{u}_i$, where the value of $C$ will be specified later. 
\end{rem}

In fact, we  consider a subset of the packing set $\packval{C/\sqrt{2}}$, defined as follows:
\begin{lemma} \label{lem:packsubset}
Let $\delta_M \in [\frac{1}{M}, 1-\frac{1}{M} ]$, where $M= (n/k)^{k/4}$.  Then there exists a subset $\packvalp{C/\sqrt{2}} \subseteq \packval{C/\sqrt{2}}$ such that  $M':=\lceil \delta_M M \rceil$ and 
\[  \max_{\theta_i \in  \packvalp{C/\sqrt{2}}} \,  \norm{\mbf{A} \theta_i}^2 \leq  \frac{\norm{\mbf{A}}_F^2  C^2 (1+\beta)}{n (1-\delta_M) }. \]
\end{lemma}
\begin{proof}
We first bound the average over the packing set $\packval{C/\sqrt{2}}$, given by    $\frac{1}{M} \sum_{i =1}^M  \| \mbf A \theta_i  \|^2$. Using steps similar to those in \cite[p.320]{candesDav13}, we have
\begin{align}
\frac{1}{M} \sum_{i =1}^M  \| \mbf A \theta_i  \|^2   
& = \frac{1}{M} \sum_{i =1}^M  \trace \left( \mbf{A} \theta_i \theta_i^T \mbf{A}^T\right)  \nonumber \\
& =   \trace \left( (\mathbf{A^T} \mathbf{A}) \frac{1}{M} \sum_{i =1}^M \theta_i \theta_i^T\right)  \nonumber \\
& \stackrel{(a)}{\leq}  \trace \left( \mathbf{A^T} \mathbf{A} \right) \norm{\frac{C^2}{M} \sum_{i =1}^M \vc{u}_i \vc{u}_i^T} \nonumber \\
& \stackrel{(b)}{\leq}   \norm{\mbf{A}}_F^2 C^2 \frac{ (1+\beta)}{n}.
\label{eq:avg_Ax_bnd}
\end{align}
In the above chain, step $(a)$ holds because both $(\mathbf{A^T} \mathbf{A})$ and $\sum_{i =1}^M \vc{u}_i \vc{u}_i^T/M$ are positive semi-definite. Step $(b)$  is obtained using Property \ref{it:packprop}  of the packing set as defined in Remark \ref{rem:candesset}.

We  use the  fact that if  the average of $M$ non-negative numbers $c_1 \leq c_2 \ldots \leq c_M$ is $c$, then $c_{j} \leq \frac{c}{1-(j-1)/M}$, for  $1 \leq j  \leq M$
(because otherwise the sum of the $(M-j+1)$ largest numbers will exceed $Mc$).  The result then follows by picking $M'$ elements of $\packval{C/\sqrt{2}}$ in increasing order of $\norm{\mbf{A} \theta}^2$, and calling this set $\packvalp{C/\sqrt{2}}$.
\end{proof}

As we restrict attention to the subset $\packvalp{C/\sqrt{2}}$ in the rest of this section, with mild abuse of notation, let us denote its elements by $\{\theta_1, \ldots, \theta_{M'} \}$. Also, let $\phi( \vc{y}; \boldsymbol{m}, \boldsymbol{\Sigma})$ denote the normal density in $\mathbb{R}^m$ with mean vector $\boldsymbol{m}$ and covariance matrix $\boldsymbol{\Sigma}$.  Then, with $\mu$ denoting the Lebesgue measure on $\mc{Y}$,  from the measurement model \eqref{eq:cs_model}, for any $\theta_i$ we have
\be
\frac{dP_{\vc{Y}|\theta_i}}{d \mu}( \vc{y} ) = \phi( \vc{y}; \mbf{A} \theta_i, \sigma^2 \mathbf{I}_m)  =  \prod_{r=1}^m  \phi(y_r; \mbf{A}^T_r \theta_i, \sigma^2 ), \label{eq:pyx_prod}
\ee
where the $r$th row of $\mbf{A}$ is denoted by  $\mbf{A}^T_r \in \mathbb{R}^n$. Further, we choose 
\be \frac{dQ_{\vc{Y}}}{d \mu}( \vc{y})=  \phi(\vc{y}; \vc{0}, \sigma^2 \mathbf{I}_m)  = \prod_{r=1}^m  \phi(y_r;  0, \sigma^2),  \label{eq:qy_prod}\ee
and  prove the following bound for the integral in \eqref{eq:main}.
 
\begin{lemma} Let $P_{\vc{Y}|\theta_i}$ and $Q_{\vc{Y}}$ given by \eqref{eq:pyx_prod} and \eqref{eq:qy_prod}, respectively.  Then for any $\lambda >0$,   
\begin{equation}
 \int_{\mc{Y}} \left( \frac{dP_{\vc{Y}|\theta_i} }{dQ_{\vc{Y}}} (\vc{y}) \right)^{1+\lambda} dQ_{\vc{Y}}(\vc{y}) 
  =
  \exp \left(  \frac{ \lambda( 1+\lambda)}{2 \sigma^2}   \| \mbf A  \theta_i  \|^2 \right).
  \label{eq:pq_exp_bnd}
\end{equation}
Hence,  the  bracketed term in \eqref{eq:main} can be bounded by 
\be  \sum_{i =1}^{M'}  \frac{1}{M'} 
 \int_{\mc{Y}} \left( \frac{dP_{\vc{Y}| \theta_i} }{dQ_{\vc{Y}}} (\vc{y}) \right)^{1+\lambda}  dQ_{\vc{Y}}(\vc{y})  \leq 
 \exp\left( \frac{\lambda( 1+\lambda)}{2 \sigma^2}  \frac{\norm{\mbf{A}}_F^2  C^2 (1+\beta)}{n (1-\delta_M)} \right).  \label{eq:subset_av}\ee
\label{lem:simplify_term}
\end{lemma}
\begin{proof} See Appendix \ref{sec:prooflemsimp}. \end{proof}

Combining Lemma \ref{lem:simplify_term} with Theorem \ref{thm:main}, we deduce the following lower bound.
\begin{prop} \label{prop:csbd}
 For any $\lambda> 0$, $\Delta \in (0,1)$, and $M=(n/k)^{k/4}$, we have 
\be \begin{split}   \minimaxcs & = \inf_{ \whx} \sup_{x \in \Sigma_k} \ep \left[ \frac{1}{n} \norm{  \whx(\vc{y}) - \vc{x}}^2 \right] \\
& \geq \frac{ \sigma^2 }{ 4 \norm{\mbf{A}}_F^2  } \left( \frac{k}{4} \log \frac{n}{k} -1 \right)  \frac{ (1-\Delta)}{(1+\lambda) (1+\beta) } \e_M,  \label{eq:fin_samp_CS} 
\end{split} \ee
where 
\be
 \e_M \geq 1  -  (1+\lambda) \left( \frac{(\log M) M^{-\Delta}}{\lambda} \right)^{\lambda/(1+\lambda)}.
 \label{eq:eM_CS}
  \ee
Therefore, for large  $n$ we have
\be \minimaxcs \geq \frac{ \sigma^2 }{ 4 \norm{\mbf{A}}_F^2 } \left( \frac{k}{4} \log \frac{n}{k}\right)(1 - o(1)).  \label{eq:CS_asymp} \ee
\end{prop}
\begin{proof}
To apply Theorem \ref{thm:main}, we equate the minimum distance  $C/\sqrt{2}$ of the packing subset $\packvalp{C/\sqrt{2}}'$ to  $2A\psi_n$. Taking $A=1$ gives $\psi_n = \frac{C}{2\sqrt{2}}$. Then, taking $w(t) =t^2$, we deduce that
\begin{eqnarray*}
\inf_{ \whx} \sup_{\vc{x} \in \Sigma_k} \ep \left[  \norm{  \whx(\vc{y}) - \vc{x}}^2 \right]
\geq  \left( \frac{C}{2 \sqrt{2}} \right)^2 \e_M.
\end{eqnarray*}
We can bound $\e_M$ by using \eqref{eq:subset_av} of Lemma \ref{lem:simplify_term} in Theorem \ref{thm:main}:
\begin{align*}
 \e_M & \geq  1 - \frac{ (1+\lambda)}{ \left( \lambda M'   \right)^{\lambda/(1+\lambda)}}
 \exp \left( \frac{ \lambda \| \mbf{A} \|_F^2 C^2 (1+\beta)}{
2 n \sigma^2(1-\delta_M) } \right) \\
& \geq  1 - \frac{ (1+\lambda)}{ \left( \lambda \delta_M   \right)^{\lambda/(1+\lambda)}}
 \exp \left( \lambda \left( \frac{ \| \mbf{A} \|_F^2 C^2 (1+\beta)}{
2 n \sigma^2(1-\delta_M) }- \frac{ \log M}{1+\lambda}  \right)   \right)
\end{align*}
Hence for any fixed $\lambda$ we obtain \eqref{eq:fin_samp_CS} and  \eqref{eq:eM_CS} by choosing
$$ C^2 =  \frac{2 n \sigma^2(1-\delta_M) \log M}{\norm{\mbf{A}}_F^2 (1+\beta)  (1+\lambda)} (1-\Delta),$$
with  $\delta_M = 1/\log M$ and $\Delta \in (0,1)$.
To obtain \eqref{eq:CS_asymp}, we recall from Remark \ref{rem:candesset} that $\beta$ can be chosen arbitrarily small as $n \to \infty$. Furthermore, $\lambda, \Delta$ can also be arranged to go $0$ (at suitably slow rates) as $n \to \infty$.
\end{proof}

\begin{rem} \label{rem:CS_comparison}  The paper \cite{candesDav13} uses Fano's inequality to derive the following bound:
\[  \mathsf{M}^*(\mbf{A}) \geq  \frac{\sigma^2} {32 \norm{\mbf{A}}_F^2(1+\beta)} \left( \frac{k}{4}\log \frac{n}{k} - 2 \right). \]
Comparing with Proposition \ref{prop:csbd}, we see that our result improves the bound by a  factor close to 8 for large $n$.
\end{rem}
\appendix
\section{Proof of Lemma \ref{lem:ch_hyp2}} \label{sec:proofch_hyp2}
For $(X,Y) \in  \mc{X} \times \mc{Y}$ consider  hypotheses 
$H_0:\; (X,Y) \sim Q$ and
$H_1:\; (X,Y) \sim P$, where 
we assume that $P \ll Q$ so that the Radon-Nikodym derivative $\frac{dP}{dQ}$ exists.

The following lemma can be found in \cite[Lemma 12.2]{pwLectNotes} (also  \cite[eq. (102)]{ppv10}).
\begin{lemma} 
For any randomized test $T$ to distinguish between the above hypotheses, and $\gamma >0$, we have 
\be P[T=1] - \gamma Q[T=1] \leq P\left[ \frac{dP}{dQ} > \gamma \right].  \label{eq:Egam_eq} \ee
\label{lem:simp_ineq}
\end{lemma}
We note that the maximum  of the left  hand side of  \eqref{eq:Egam_eq} (over all all tests $T$) is the $E_\gamma$ divergence \cite{liu2017e_gam,sasonV16}. We use \eqref{eq:Egam_eq} to complete the proof of Lemma \ref{lem:ch_hyp2}:
\begin{proof}[Proof of Lemma \ref{lem:ch_hyp2}]
As in \cite[Theorem 26]{ppv10}, let $\e_M$ and $\e'_M$ denote the average error probabilities over channels $P_{\vc{Y}|\theta}$ and $Q_{\vc{Y}|\theta} = Q_{\vc{Y}} $, respectively, 
 for a channel code with $M$ equiprobable codewords. 
Given $(\theta,\vc{Y})$, the result \cite[Theorem 26]{ppv10} describes a (sub-optimal) hypothesis test based on the channel decoder to distinguish between $H_0 : (\theta, \vc{Y}) \sim Q_{\theta\vc{Y}} = \pi_\theta Q_{\vc{Y}}$ and $H_1:  (\theta, \vc{Y}) \sim P_{\theta\vc{Y}} =  \pi_\theta P_{\vc{Y}|\theta}$. Let $T \in \{0, 1 \}$ denote the output of this test.
It is shown in the proof of that theorem that the probability of Type I error, i.e.,  $Q[T=1]$ is $1-\e'_M$, and the probability of type II error, i.e., $P[T=0] = \e_M$. Applying Lemma \ref{lem:simp_ineq} to this hypothesis test yields that for any $\gamma >0$,
 \be
 1- \e'_M \geq \frac{1}{\gamma} \left(1- \e_M -  P_{\theta\vc{Y}} \left[ \frac{dP}{dQ}  > \gamma \right]  \right). 
\label{eq:app_e'}
 \ee

We observe that  
 when  $Q_{\vc{Y}|\theta} = Q_{\vc{Y}} $, any channel decoder has  average error probability $\e'_M =\frac{M-1}{M}$. The result in
\eqref{eq:M_lb2} follows by substituting this value for $\e'_M$ in \eqref{eq:app_e'}.
\end{proof}

\section{Recovering Fano's Inequality from  Theorem \ref{thm:main}}  \label{app:Fano}

Here we show how to obtain Fano's inequality from  Theorem \ref{thm:main}.  
We first establish a general converse result  involving mutual information (equation \eqref{eq:lnM_bnd}), and then obtain  Fano's inequality from it.

From the variational representation  in \eqref{eq:sup_gam_bnd},  for any $\lambda, \gamma >0$ we have
\begin{eqnarray}
\frac{1}{M} & \geq & \  \frac{(1-\e_M)}{\gamma}  - \frac{1}{\gamma^{1+\lambda}} \sum_{i =1}^M  \frac{1}{M}  \int_{\mc{Y}} 
\left( \frac{dP_{\vc{Y}|\theta_i} }{dQ_{\vc{Y}}} (\vc{y}) \right)^{1+\lambda}  dQ_{\vc{Y}}(\vc{y}) \nonumber \\
& = & \  \frac{(1-\e_M)}{\gamma}  - \frac{1}{\gamma^{1+\lambda}} \left(
\lambda \mc{H}_{1+\lambda}(P_{\theta \vc{Y}} || Q_{\theta \vc{Y}} ) + 1 \right), \label{eq:lb_weak}
\end{eqnarray}
where  $H_{1+\lambda}(P \| Q) := \frac{1}{\lambda} \int \Big( \left(\frac{dP}{dQ} \right)^{1+\lambda} -1\Big) dQ$ is the Hellinger divergence of order $(1+\lambda)$  from distribution $P$ to distribution $Q$. We note from \eqref{eq:Renyi_div_def} that the R\'enyi and Hellinger divergences of order $(1+\lambda)$ are invertible functions of one another. We  use the following bound  \cite[Theorem 8]{sasonV16} for the Hellinger divergence:
\be
 \mc{H}_{1+\lambda}(P \| Q) 
\leq  \kappa(\lambda,t) D(P \| Q),
\label{eq:HDPQ_bnd}
\ee
where
\be
t: = \esssup \, \frac{dP}{dQ}(x,y), \mbox{ for $(x,y) \sim Q$} \mbox{\quad and \quad} 
\kappa(\lambda,t) =  \frac{\lambda + t^{1+\lambda} - (1+\lambda)t}{\lambda (t \log t + 1 -t)}.
\label{eq:tdef}
\ee
We choose
\be Q_{\vc{Y}} = \ol{P}_{\vc{Y}} = \sum_{j=1}^M \frac{1}{M}P_{ \vc{Y} |\theta_j} . \label{eq:qydef0} \ee 
With this $Q_{\vc{Y}}$, we have that $t \leq M$ since 
$\ol{P}_{\vc{Y}}(\mc{A}) \geq M^{-1} P_{\vc{Y}|\theta_j}(\mc{A})$ for all measurable sets $\mc{A}$. We also  have
\be D(P_{\theta\vc{Y}}||Q_{\theta\vc{Y}}) = I(\theta; \vc{Y})=   \sum_{i=1}^M \frac{1}{M} D(P_{\vc{Y}|\theta_i} || \ol{P}_{\vc{Y}}).  \label{eq:ID_rel} \ee 
Substituting in \eqref{eq:HDPQ_bnd} and then in \eqref{eq:lb_weak}, we obtain
\ben
 \frac{1}{M} \geq \frac{(1-\e_M)}{\gamma}  - \frac{1}{\gamma^{1+\lambda}} \left( \lambda \kappa(\lambda,t) I(\theta; \vc{Y})  +1\right).
\een
Maximizing  over $\gamma >0$ yields
\ben
M \leq \frac{(1 + \lambda)^{1+1/\lambda}}{\lambda\,  (1-\e_M)^{1+1/\lambda}} \left[ \lambda \kappa(\lambda,t) I(\theta; \vc{Y}) +1  \right]^{1/\lambda}.
\een
Taking logs, for any $\lambda >0$ we have
\be
\log M \leq \left( 1 + \frac{1}{\lambda} \right) \log \frac{1+\lambda}{1-\e_M} - \log \lambda + \frac{1}{\lambda}\log(1+ c(\lambda,t) I(\theta; \vc{Y})),
\label{eq:lnM_bnd}
\ee
where 
\begin{equation} \label{eq:cuppbd}
 c(\lambda, t) = \lambda \kappa(\lambda, t) = \frac{\lambda + t^{1+\lambda} - (1+\lambda)t}{t \log t + 1 -t} \leq \frac{t^{\lambda}}{\log t  -1}   \mbox{\;\;\; if $t \geq e$,}
\end{equation}
where the final inequality follows by direct comparison.

Hence, for a fixed $\lambda >0$ and $M \geq 3$, using \eqref{eq:cuppbd} and $t \leq M$  in  \eqref{eq:lnM_bnd} gives
\ben
0 \leq (1+\lambda) \log \frac{1+\lambda}{1-\e_M} - \lambda \log \lambda + \log \left(  M^{-\lambda} + \frac{I(\theta;\vc{Y})}{\log M - 1}  \right).
 \een
 Or,
\be \log M \leq 1 + I(\theta; \vc{Y}) \left( \frac{\lambda^\lambda (1-\e_M)^{1+\lambda} }{ (1+\lambda)^{1+\lambda} } - M^{-\lambda} \right)^{-1}.  \label{eq:lnMlamb} \ee
 Finally, noting that $\lambda$ can be chosen arbitrarily small, choose $\lambda = 1/(\log M)^{\alpha}$ for some $\alpha \in (0,1)$ in  \eqref{eq:lnMlamb}. We therefore have
 \ben \log M \leq 1+  \frac{I(\theta; \mbf{Y}) }{1-\e_M}(1+ o(1)).  \een
 Using the expression for $I(\theta; \mbf{Y})$ in \eqref{eq:ID_rel} and rearranging, we get
 \[ \e_M \geq  1-  \frac{ \frac{1}{M} \sum_{i=1}^M  D(P_{\theta_i}  || \ol{P}) }{\log M -1}(1+ o(1)),  \]
 where $o(1)$ denotes a term that goes to zero with growing $M$. We have thus recovered Fano's inequality in \eqref{eq:Fano} to within $o(1)$ terms.

\section{Proof of Lemma \ref{lem:explicit}} \label{sec:prooflemexmplicit}

\begin{proof}[Proof of Lemma \ref{lem:explicit}]
Since $Q_{\vc{Y}}$ is the uniform measure and each $\frac{dP_{\vc{Y}|\theta_i}}{dQ_{\vc{Y}}}$ corresponds to an $\ftau^n$, 
for each value of $i$, we can express the relevant integral  as
\begin{equation} \label{eq:tocontrol3}
 \int_{\mc{Y}} \left( \frac{dP_{\vc{Y}|\theta_i}}{dQ_{\vc{Y}}} (\vc{y}) \right)^{2} dQ_{\vc{Y}}(\vc{y})
  = \int_{[0,1]^n} \ftau^n(\vc{y})^2 d\vc{y} = \left( \int_0^1 \ftau(y)^2 dy \right)^n.
\end{equation}
For any $\vd{\tau}$ we can express the bracketed term on the RHS of \eqref{eq:tocontrol3} as
\begin{align}
\int_0^1 \ftau(y)^2 dy & =   \int_0^1  \left( 1 + \sum_{j=0}^{m-1} \tau_j g_j(y) \right)^2 dy \nonumber \\
& =   1 + 2 \sum_{j=0}^{m-1} \tau_j \int_0^1 g_j(y) dy + \sum_{j=0}^{m-1} \sum_{k=0}^{m-1} \tau_j \tau_k \int_0^1 g_j(y) g_k(y)  dy  
\nonumber \\
&\stackrel{(a)}{=}     1  + \frac{c^2 a}{m^4}. \label{eq:tocontrol2}
\end{align}
Here equality $(a)$ is obtained from \eqref{eq:gprop2} and  \eqref{eq:gjdef} since  $g_j(y) g_k(y) \equiv 0$ for $j \neq k$, and
\begin{align*}
\int_0^1 g_j(y) dy & = \frac{c}{m^2} \int_{j/m}^{(j+1)/m} g(m y - j) dy
= \frac{c}{m^2} \int_0^1 g(u) \frac{du}{m} = 0, \\
\int_0^1 g_j(y)^2 dy & = \frac{c^2}{m^4} \int_{j/m}^{(j+1)/m} g(m y - j)^2 dy
= \frac{c^2}{m^4} \int_0^1 g(u)^2 \frac{du}{m} = \frac{c^2 a}{m^5}. 
\end{align*}
The result follows on substituting \eqref{eq:tocontrol2} into \eqref{eq:tocontrol3} and using $(1+x)^n \leq \exp(x n)$ for any $x \in \re$.
\end{proof}

\section{Proof of Lemma \ref{lem:classmain}} \label{sec:proofclassmain}
\begin{proof}[Proof of Lemma \ref{lem:classmain}]
We can express the key ratio on the LHS of \eqref{eq:classmain} as
\begin{equation}
 \frac{dP_{\vd{\tau} }}{ dQ_{\vc{U},\vc{V}}} (\vc{U},\vc{V}) = \prod_{r=1}^n
 \frac{  dP_{\vd{\tau}} }{ dP_{\vc{0}}} (V_r|U_r)  = \prod_{r=1}^n
 \frac{  P_{\vd{\tau}}(V_r|U_r) }{ P_{\vc{0}}(V_r|U_r) }.
 \label{eq:dPQUV0}
\end{equation}
We note that $dP_{\vc{0}} (V_r|U_r) =  P_{\vc{0}} (V_r|U_r) d \nu(V_r)$, where $\nu$ represents the counting measure on $\{0,1\}$. Then,  using \eqref{eq:dPQUV0} we write the integral in \eqref{eq:classmain} as
\begin{eqnarray*}
\lefteqn{
\int_{\mc{Y}_n } 
\left( \prod_{r=1}^n
 \frac{  dP_{\vd{\tau}}}{ dP_{\vc{0}} } (v_r|u_r)  \right)^{1+\lambda}
\prod_{r=1}^n dP(u_r | (u,u)_r^-) \prod_{r=1}^n dP_{\vc{0}}( v_r |u_r) } \\
& = &  \int_{\mc{Y}_{n-1} } \;\;  \prod_{r=1}^{n-1} \left(
 \frac{  dP_{\vd{\tau}} }{ dP_{\vc{0}} } (v_r|u_r)  \right)^{1+\lambda} \biggl[ I_n \biggr]
\prod_{r=1}^{n-1} dP(u_r | (u,u)_r^-) \prod_{r=1}^{n-1} dP_{\vc{0}}( v_r |u_r) 
\end{eqnarray*}
where the inner integral $I_n$ can be written as
\begin{align}
    I_n & :=  \int_{[0,1]^d} \left( \frac{ [\eta_{\vd{\tau}}( u_n)]^{1+\lambda} }{ [\eta_{\vc{0}}(u_n)]^{\lambda}}  
+ \frac{[1-\eta_{\vd{\tau}}( u_n)]^{1+\lambda} }{[1- \eta_{\vc{0}}( u_n)]^{\lambda}}  \right) dP_{U_n | (U,V)_n^-}(u_n | (u,v)_n^-)  \nonumber  \\
& =  \int_{[0,1]^d}  \exp\left( \lambda D_{1+\lambda} \left( P_{\vd{\tau}}( \cdot | u_n) ||  P_{0}( \cdot | u_n) \right) \right)  dP_{U_n | (U,V)_n^-}(u_n | (u,v)_n^-) \label{eq:simple2},
\end{align} 
where we use the fact that the R\'{e}nyi divergence of order $(1+\lambda)$ between two Bernoulli random variables with parameters $ \eta_{\vd{\tau}}(u_n)$ and $\eta_{\vc{0}}(u_n)$, respectively, is 
\ben
D_{1+\lambda} \left( P_{\vd{\tau}}( \cdot | u_n) ||  P_{0}( \cdot | u_n) \right) = \frac{1}{\lambda} 
\log \left(  \frac{ [\eta_{\vd{\tau}}( u_n)]^{1+\lambda} }{ [\eta_{\vc{0}}(u_n)]^{\lambda}}  
+ \frac{[1-\eta_{\vd{\tau}}( u_n)]^{1+\lambda} }{[1- \eta_{\vc{0}}( u_n)]^{\lambda}} \right).
\een
Recalling that $\beta_m=LHm^{-\alpha}$ and $u_n \in [0,1]^d$, let us denote the $d$th coordinate of $u_n$ by $u_{n,d}$.
The construction of the hypercube class of functions $\FF_m$ in \cite[p.2350]{castroNowak08} ensures that  for any $\vd{\tau}, \vd{\tau'} \in \{0,1 \}^{m^{d-1}}$,  the following properties hold.
\begin{align}
 \eta_{\vd{\tau}}(u_n)  & = \eta_{\vd{\tau'}}(u_n),  \quad  \beta_m \leq u_{n,d} \leq 1, \\
 \frac{1}{2} - c \beta_m \leq \eta_{\vd{\tau}}(u_n) & \leq  \frac{1}{2} +  c \beta_m, \quad 0 \leq u_{n,d} \leq \beta_m, \label{eq:eta_UpLo} \\
 \abs{\eta_{\vd{\tau}}(u_n)  - \eta_{\vd{\tau'}}(u_n)} & \leq 2c \beta_m^{\kappa-1},   \quad \forall \, u_n \in [0,1]^{d}. \label{eq:eta_sep}
\end{align}
 
We now use the following bound on the R\'{e}nyi divergence due to Verd{\'u} and Sason \cite[Theorem 3]{sasonV15upper} for $u_{n,d} \leq \beta_m$ and $\lambda \in [0,1]$:
\be
D_{1+\lambda} \left( P_{\vd{\tau}}( \cdot | u_n) ||  P_{0}( \cdot | u_n) \right)  \leq \log\left(1 + \frac{2 \delta^2}{\min_{v \in \{0,1 \}} \, P_{0}( v | u_n)} \right),
\label{eq:D1L_ub}
\ee
where $\delta := \abs{\eta_{\vd{\tau}}(u_n) - \eta_{\vc{0}}(u_n)}$ is the total variation distance between $P_{\vd{\tau}}( \cdot | u_n)$ and   $P_{0}( \cdot | u_n)$.
Using  \eqref{eq:eta_sep} for an upper bound on $\delta$, and \eqref{eq:eta_UpLo} for a lower bound on the minimum of $P_{0}( \cdot | u_n)$, we have from \eqref{eq:D1L_ub},
\ben
D_{1+\lambda} \left( P_{\vd{\tau}}( \cdot | u_n) ||  P_{0}( \cdot | u_n) \right)  \leq \log \left(1 + \frac{8c^2 \beta_m^{2(\kappa-1)}}{\frac{1}{2} - c \beta_m }\right).
\een
Substituting this bound in \eqref{eq:simple2} to bound $I_n$, we obtain using $1+x \leq e^x$ that 
\[  I_n \leq  \left(1 + \frac{8 c^2 \beta_m^{2(\kappa-1)}}{\frac{1}{2} - c \beta_m} \right)^{\lambda} \leq \exp \left( \frac{16 c^2 \beta_m^{2(\kappa-1)} \lambda}{1 - 2c \beta_m } \right).  \]
The result follows by induction on $n$.
\end{proof}

\section{Proof of Lemma \ref{lem:simplify_term}} \label{sec:prooflemsimp}

Recall from \eqref{eq:pyx_prod} and \eqref{eq:qy_prod} that we take
$dP_{\theta_i}(\vc{y} | \theta_i)   =  \prod_{r=1}^m  \phi(y_r; \mbf{A}^T_r \theta_i, \sigma^2 )$ and
$dQ_{\vc{Y}}(\vc{y})  = \prod_{r=1}^m  \phi(y_r;  0, \sigma^2)$.

\begin{proof}[Proof of Lemma \ref{lem:simplify_term}]
For any $\lambda > 0$, we have
\begin{align*}
  \int_{\mc{Y}}  \left( \frac{dP_{\vc{Y}| \theta_i} }{dQ_{\vc{Y}}} (\vc{y}) \right)^{1+\lambda} dQ_{\vc{Y}}(\vc{y}) 
& =  \prod_{r=1}^m  \int_{\re} \frac{[\phi(y_r; \mbf{A}^T_r \theta_i, \sigma^2) ]^{1+\lambda}}
{[\phi(y_r;  0,  \sigma^2)]^\lambda} dy_r  \\
& \stackrel{(a)}{=}     
 \prod_{r=1}^m \exp \left(  \frac{ \lambda( 1+\lambda)}{2 \sigma^2}   \left( \mbf{A}^T_r \theta_i  \right)  ^2  \right) \\
& =   
  \exp \left(   \frac{\lambda( 1+\lambda)}{2 \sigma^2}   \| \mbf A \theta_i  \|^2 \right).
\end{align*}
The equality in step $(a)$ is obtained by completing the square
inside an exponential, and recognizing the remaining term as a multiple of a normal density.

To obtain \eqref{eq:subset_av}, we use Lemma \ref{lem:packsubset} to bound $\norm{\mbf{A} \theta_i}^2$ on the RHS of  \eqref{eq:pq_exp_bnd} for each $\theta_i \in \packvalp{C/\sqrt{2}}$.
\end{proof}
\subsection*{Acknowledgements}
 The authors  thank B. Nakibo{\u{g}}lu for pointing out the optimal choice of $Q_{\vc{Y}}$ given in Remark \ref{rem:optQ}, and Igal Sason for helpful comments. They also thank the Alan Turing Institute for funding to attend the scoping workshop `Information-Theoretic Foundations and Algorithms for Large-Scale Data Inference'. RV would like to acknowledge support from a Marie Curie Career Integration Grant (Grant Number 631489).

\bibliographystyle{ieeetr}
{\small \bibliography{references}}
\end{document}